\documentclass[a4paper, 12pt]{article}

\usepackage[sort&compress]{natbib}

\usepackage{amsthm, amsmath, amssymb, mathrsfs, multirow, enumitem}
\usepackage[hyphens]{url}
\usepackage{graphicx, subfigure} 
\usepackage{ifthen} 
\usepackage{amsfonts}
\usepackage[usenames]{color}
\usepackage{fullpage}
\usepackage{setspace, lineno}

\theoremstyle{plain}

\newtheorem{prop}{Proposition}

\theoremstyle{definition}

\theoremstyle{remark}

\newcommand{\prob}{\mathsf{P}} 
\newcommand{\E}{\mathsf{E}}

\newcommand{\probset}{\mathbb{P}}

\newcommand{\bin}{{\sf Bin}}

\newcommand{\YY}{\mathbb{Y}}

\def\FTE{\mathop\mathit{ FiveThirtyEight}\nolimits}

\newcommand{\clinton}{\text{\sc c}}
\newcommand{\trump}{\text{\sc t}}

\newcommand{\upi}{\overline{\pi}}
\newcommand{\lpi}{\underline{\pi}}

\renewcommand{\phi}{\varphi}

\title{Rethinking probabilistic prediction in the wake of the 2016 U.S.\ presidential election}
\author{
Harry Crane\footnote{Department of Statistics \& Biostatistics, Rutgers, the State University of New Jersey, {\tt hcrane@stat.rutgers.edu}} 
\quad and \quad 
Ryan Martin\footnote{Department of Statistics, North Carolina State University, {\tt rgmarti3@ncsu.edu}}
}
\date{\today}

\begin{document}

\maketitle 

\begin{abstract}   
To many statisticians and citizens, the outcome of the most recent U.S.\ presidential election represents a failure of data-driven methods on the grandest scale.  This impression has led to much debate and discussion about how the election predictions went awry---Were the polls inaccurate? Were the models wrong? Did we misinterpret the probabilities?---and how they went right---Perhaps the analyses were correct even though the predictions were wrong, that's just the nature of probabilistic forecasting. With this in mind, we analyze the election outcome with respect to a core set of {\em effectiveness principles}.  Regardless of whether and how the election predictions were right or wrong, we argue that they were ineffective in conveying the extent to which the data was informative of the outcome and the level of uncertainty in making these assessments.  Among other things, our analysis sheds light on the shortcomings of the classical interpretations of probability and its communication to consumers in the form of predictions.  We present here an alternative approach, based on a notion of {\em validity}, which offers two immediate insights for predictive inference.  First, the predictions are more conservative, arguably more realistic, and come with certain guarantees on the probability of an erroneous prediction.  Second, our approach easily and naturally reflects the (possibly substantial) uncertainty about the model by outputting {\em plausibilities} instead of {\em probabilities}.  Had these simple steps been taken by the popular prediction outlets, 
the election outcome may not have been so shocking.

\smallskip

\end{abstract}

\section{Introduction}

On the morning of November 8, 2016, the day of the United States presidential election, Nate Silver's $\FTE$ website estimated Hillary Clinton's chance of being elected President at 72\%, {\em The New York Times} estimated Clinton at 91\%, {\em The Huffington Post} 98\%, and the Princeton Election Consortium  99\% \cite{HuffPo,NYTimes1}. 
 Bolstered and perhaps influenced by these probabilities, the analyses of many political pundits pointed to a likely, if not inevitable, Clinton victory.
Yet, at 3am the next morning, Donald Trump delivered a victory speech as the President-Elect of the United States.

In the days and weeks that followed, discussion of the historic election quickly turned into debate over what went wrong and who or what was to blame:
\begin{itemize}\addtolength{\itemsep}{-0.3\baselineskip}
\item some blamed the data, including Sam Wang, co-founder of the Princeton Election Consortium, who told {\em The New York Times} that ``{state polls were off in a way that has not been seen in previous presidential election years}'' \cite{NYTimes1};
\item some changed their story, including $\FTE$, who reframed the narrative from ``{there's a wide range of outcomes, and most of them come up Clinton}'' on the day of the election \cite{Silver1108}
to the assertion that there were ``{plenty of paths to victory for Trump}'' and ``{people should have been better prepared for it}'' afterwards \cite{Silver1111};
\item others were in denial, including Irizarry, who was bold enough to declare 
that ``{statistically speaking, Nate Silver, once again, got it right}'' \cite{SS}.
\end{itemize}

Based on this fallout, $\FTE$, {\em The New York Times}, Princeton Election Consortium, and many newcomers will likely put forth ``new and improved'' methods for the 2020 election.  Time will tell how improved those methods turn out to be.
In the meantime, however, the fields of statistics and data science will continue on much like they did before the election.  
For the field at large, the specific methodologies undertaken to obtain the predictions of 72\%, 91\%, 98\%, 99\% for Clinton are of little significance to the many more salient and consequential applications of statistical methods in science, medicine, finance, etc.  With the interests of broader statistical practice in mind, we dispense with any attempts to diagnose in what ways the election analyses were ``right'' or ``wrong.''  We instead treat the election analysis meta-statistically, extracting three main insights and highlighting pressing questions
about how probabilistic predictions ought to be conducted and communicated moving forward.

First, we address the tendency, even among statisticians and data scientists, to be results-oriented in assessing the effectiveness of their analyses.  
The widely publicized success of $\FTE$'s 2008 and 2012 election predictions coupled with larger societal trends toward quantitative analysis fuel a kind of blind faith in data-driven methods.
Trends toward data-based methods across the spectrum of business, industry, and politics have been accompanied by a shift away from principles in statistical circles. 
As Speed stated plainly in a recent column of the {\em IMS Bulletin}:
\begin{quote}
\label{quote:speed}
{\em If a statistical analysis is clearly shown to be effective [...]~it gains nothing from being described as principled} \cite{Speed}.
\end{quote}
But if the effectiveness of a certain method is primarily judged by the fact that it seems to have worked in hindsight, then what are we left with when such a method fails?  In the wake of the 2016 election, we are left with outrage, confusion, questions, and not many answers.  Moreover, since the effectiveness of any given statistical analysis must be assessed prior to observing the outcome, such a judgment cannot be entirely results-driven.  It must be based on principles.  In  Section~\ref{S:effective} we present three basic {\em effectiveness principles} for any statistical analysis and discuss these principles in the context of the 2016 election.

Second, we discuss perhaps the most mysterious aspect of election analysis and probabilistic analyses in general: the probabilities themselves. 
Is the $\FTE$ probability of $72\%$ for Clinton to be interpreted as a frequency, a degree of belief, a betting quotient, or something else?  Are we to assess the quality of these probabilities in terms of {coherence}, calibration, proper scoring, or some other way?  Do these differences matter?
Even among statisticians there is widespread disagreement on how to answer these questions.  Worse is the impulse by some to dismiss questions about the meaning of probability and the quality of predictions.  
\begin{quotation} \label{quote:mit.guy}
{\em ``[T]he key thing to understand is that [data science] is a tool that is not necessarily going to give you answers, but probabilities,'' said Erik Brynjolfsson, a professor at the Sloan School of Management at the Massachusetts Institute of Technology.

Mr.~Brynjolfsson said that people often do not understand that if the chance that something will happen is 70 percent, that means there is a 30 percent chance it will not occur. The election performance, he said, is ``not really a shock to data science and statistics. It's how it works''} \cite{NYTimes1}.
\end{quotation}
In high-impact and high-exposure applications like this election, when the reputation of data science is on the line, the interpretation and communication of probabilities to the general public ought to be taken more seriously than Brynjolfsson suggests.  After discussing the merits and demerits of the classical interpretations of probability for election predictions in Section~\ref{S:valid}, 
we propose a new approach that avoids the various shortcomings of traditional perspectives by only requiring an interpretation of small probability values, \`a la Cournot \cite{Cournot}.  A key feature of this perspective is that the derived predictions are {\em valid}, i.e., the probability of an erroneous prediction is no more than a specified threshold.  This validity property is achieved by making more conservative predictions, thereby avoiding the charge of ``overselling precision'' 
raised by some after the election \cite{NYTimes1}.

Third, we make the important distinction, arising out of the discussion of Section~\ref{S:effective}, between {\em variability} in statistical estimates and  {\em uncertainty} in the model setup.  While it is common to report the former, it is uncommon to report the latter when making predictions.  More precisely, a statistical model encodes uncertainty in the data, 
which can be used to derive measures of variability for estimators, confidence intervals, etc.  But these measures are only reliable insofar as the model is sound for a given application.  They do not account for variability in estimators due to uncertainty in the choice of model.
In Section~\ref{S:uncertainty}, we describe how to account for this latter type of uncertainty in the statistical analysis.  Specifically, we propose an alternative output of {\em plausibility} rather than {\em probability}, which accounts for uncertainty  in a natural way that makes predictions more conservative and assessments of predictive validity more stringent.

Finally, in Section~\ref{S:fail}, we conclude with some remarks about the statistical community's reaction to the election outcomes, and how the perceived failure of data science in the election predictions might be turned into an opportunity to better communicate the capabilities and limitations of statistical methods.

\section{Were the analyses effective?}
\label{S:effective}

We begin with the premise that statistical analyses ought to be evaluated based on how {\em effective} they are in providing insight about a scientific question of interest, not whether they produce {\em correct} predictions.
In election analysis, the question is simple: {\em Which candidate will win the election?}
The analyses of several outlets, including $\FTE$, Princeton Election Consortium, etc., were conducted with this question in mind.  Were those analyses effective?  

We first stress that in many statistical applications, the correct answer or outcome may never be revealed, making a right/wrong classification of predictions impossible.
Even in cases like the election, where the outcome is eventually revealed, an assessment of effectiveness is only useful if made {\em prior} to learning the outcome.
Effectiveness, therefore, cannot be judged by whether or not an approach ``worked'' in a particular instance, in terms of giving a ``right'' or ``wrong'' prediction. 
Similarly, a determination of effectiveness cannot be based on other statistical analyses, for then the effectiveness of the meta-analysis must be established, {\em ad infinitum}.  Effectiveness must instead  be established by non-statistical means, which include principles, subject matter knowledge, qualitative assessments, intuition, and logic.  
 In this section, we evaluate the 2016 election analyses based on a set of {\em effectiveness principles} \cite{effective}, which are necessary to ensure that the three main components of statistical practice---data, model, and inference---fit together in a coherent way.  In particular, data must be {\em relevant} for the question of interest; the model must be {\em sound} for the given data; and the inference/prediction must be {\em valid} for the given model. 

\subsubsection*{\em Principle~1: The data must be relevant}

For a statistical analysis to be effective, the data must be relevant for answering the question under consideration.  With regards to the question, {\em who will win the presidential election?}, common sense suggests that polling data is relevant to the election outcome.  However, this judgment relies on an assumption that those polled are representative of the general electorate and their responses are truthful (at least at the time the poll is taken).\footnote{For the sake of our discussion, we assume that the analyses carried out by $\FTE$, Princeton Election Consortium, {\em The New York Times}, etc., were based largely on polling data.  $\FTE$ offers a polls-only and a polls-plus prediction; the latter includes some additional information about other current trends, but these estimates rarely differ by more than a few points.}  In reality, both of these assumptions---and others from the classical survey sampling theory---may have been violated to some extent in the 2016 election polling.

The challenge to statistical analysis, however, is not the binary question of {\em relevance} or {\em irrelevance}, but rather the manner in which the data may be relevant for the given question.
Although polling data has been useful in the past and continues to be relevant in the present, the way in which the polls convey this information is likely to change over time.
For example, the advent of social media and the decline of landline telephones has changed the way in which polls are conducted, which in turn may affect (positively or negatively) the information we are able to glean from the polls.

In the 2016 election, in particular, several statisticians, including Gelman, have admitted that the relevance of the polls was incorrectly assessed by many of the analyses.  
\begin{quote}
{\em So Nate [Silver] and I were both guilty of taking the easy way out and looking at poll aggregates and not doing the work to get inside the polls. We're doing that now, in December, but I/we should've been doing it in October. Instead of obsessing about details of poll aggregation, we should've been working more closely with the raw data} \cite{Gelman}.
\end{quote}
Such comments stand in stark contrast to the excuse by Sam Wang that the ``polls were off'' and, more broadly, that
\begin{quote}
{\em In this presidential election, analysts said, the other big problem was that some state polls were wrong
} \cite{NYTimes1}.
\end{quote}

From a scientific point of view, the concept that the {\em polls were wrong} suggests a perplexing juxtaposition of the experiment and the experimenter, the science and the scientist.  The polls are the data used by $\FTE$, Princeton Election Consortium, etc., to make predictions.  And what is a data scientist's job if not to understand data?  Just as it is a biologist's job to understand the biology and a physicist's job to understand the physical world.    
Blaming a faulty statistical analysis on the data is tantamount to a physicist blaming an electron or a biologist blaming a lab rat for a flawed scientific hypothesis.  Armed with relevant data, the data scientist is tasked with gleaning whatever insights the data can provide.  Attempts by Wang and others to blame faulty predictions on erroneous polls, then, only highlight their failure to properly account for the key sources of variation in the data.  
This is not an error in the data, but rather in the modeling of that data.

\subsubsection*{\em Principle~2: The model must be sound}

The statistical model must be sound in the sense that it accounts for sources of variation, whether as a result of data collection or natural variation in the underlying system.  Claims, such as Wang's, that the {\em polls were off} are based primarily on the fact that the usual models did not fit the 2016 polls as well as they have fit past polls.  
Such reasoning is backwards.  The model assumed that the data was relevant in a way that it was not, likely because the chosen model failed to adequately account for certain sources of variability in the polling data, as Gelman points out above.

Leading up to election day, there was ample discussion of a possible ``hidden'' Trump vote,  akin to the so-called {\em Bradley effect} \cite{Bradley} from the 1982 California governor's race.    The reasoning behind a possible hidden Trump vote was based on sound, albeit non-quantitative, reasoning: quite simply, it was socially unacceptable to openly support Trump in many parts of the country and, as the theory goes, Trump voters in those regions would either say nothing or claim to support Clinton but secretly vote for Trump on election day.  This speculation was largely dismissed by data scientists because, unironically, there was no evidence for it in the data.   In the post-election analysis, $\FTE$ holds to the claim that ``shy'' voters were not the reason why the outcome differed from what the polls suggested \cite{Silver1116}, but still based on the assumption that the existence of the ``hidden'' voter would have manifested itself in the data.  Ultimately, it is unclear whether the Trump vote was actually ``hidden'' or merely overlooked by the current statistical apparatus.  What is clear, however, is that the predictors' models were consistently {\em unsound} in that they did not properly account for the variability in the polls.

The predictions also failed to reflect the {\em uncertainty} in what information the polls conveyed about the eventual outcome.  For example, rather than dismiss the possibility of a hidden Trump vote, or any other explanation for why the polls may have been skewed, the predictions would have been more insightful if they reflected the model uncertainty in their reported predictions.  While accounting for such uncertainty makes the predictions less definitive, it makes them more reliable and less sensitive to arbitrary or unjustified assumptions.  We discuss the issue further in Section~\ref{S:uncertainty}.

\subsubsection*{\em Principle~3: The prediction must be valid}

For a statistical analysis to be effective, the inference drawn must be valid, meaning that the conclusions can be justified to someone else who otherwise agrees with the model setup.  For this, it must be clear how the statistical output is to be interpreted.  In the present context, the output assigns a percentage to each candidate---e.g., 72\% for Clinton and 28\% for Trump---and inference corresponds to a prediction based on the output.  Assessing the quality of these inferences requires an understanding of how these probabilities are to be interpreted.  Therefore, the probabilities themselves do not constitute a prediction---{\em a prediction is an interpretation of the probabilities}.

For example, on October 24, 2016, Silver \cite{Silver1024} estimated Clinton's probability of victory at 83--85\%.  He interpreted this figure to mean:
\begin{quote}
{\em Hillary Clinton is probably going to become the next president}.
\end{quote}
Where does this conclusion come from?  Perhaps Silver interpreted the estimated 83\%/17\% Clinton/Trump probabilities (at that time) in terms of betting odds, making Clinton roughly a 5-to-1 favorite to win.  Or maybe he instead interprets these percentages vaguely as ``measures of risk,'' as he said in a TV interview after the election \cite{DailyShow}.
Silver is free to interpret these values as he pleases when arriving at his prediction that Clinton will ``probably'' win, but a clear interpretation is necessary if the 83\% probability is to convey any additional information in terms of how likely the prediction is to be correct.  In particular, Silver's choice of the word {\em probably} is not insignificant; it suggests some vague level of certainty, apparently stronger than ``maybe'' and weaker than ``definitely,'' which can have a significant influence on the opinions---and even actions---of the average $\FTE$ consumer.  For example, contrast Silver's prediction that Clinton will ``probably'' win based on probability 83\% on October 24th to his prediction that she will ``most likely'' win based on probability 72\% on November 8 \cite{Silver1108}.  Of course, 83\% is bigger than 17\% so, if forced to choose between the two, it makes sense that Silver would choose Clinton.  But nothing forces Silver, or any one else, to make such a choice: one could, quite reasonably, declare the race ``too close to call.'' 

Since the various outlets, $\FTE$ included, offer many different numbers without a clearly spelled out rationale for how to interpret them, it is hard to argue for the validity of any one of these predictions.  And, furthermore, it is unclear whether and how to compare probabilities from different sources.  The Princeton Election Consortium's 99\% probability conveys greater confidence in a Clinton win than {\em FiveThirtyEight}'s 72\%, but on what scale are we to make such comparisons?

\section{Valid predictions}
\label{S:valid}

Consider the pair $(X,Y)$, where $X$ represents polling data and any other information used to arrive at a prediction and $Y$ is the actual outcome of the election.  In general, $Y$ takes values in the finite set $\YY$ consisting of all eligible presidential candidates, i.e., natural-born U.S.\ citizens over the age of thirty-five who have resided within the United States for at least fourteen years as of 2016.  Much of our discussion below, however, specializes to the case in which $Y$ takes only two values, i.e., $\YY = \{\clinton, \trump\}$ is binary, Clinton ($\clinton$) or Trump ($\trump$).  The decision to disregard other candidates (Johnson, Stein, McMullen, etc.)~is itself an inference based on the fact that every conceivable model, as of early November 2016, gave these candidates a negligible win probability; see Section~\ref{SS:valid}.

Prediction probabilities are determined by a mapping $x \mapsto \pi_x$ that takes data $x$ to a probability distribution $\pi_x$ on $\YY$.   
Technically, $\pi_x$ could be any (possibly $x$-dependent) distribution, but, following our discussion in Section~\ref{S:effective}, a natural approach is to start with a model, i.e., a joint distribution $\prob$ for $(X,Y)$, and to take $\pi_x$ as the conditional distribution of $Y$, given $X=x$, under $\prob$, namely, 
\begin{equation}
\label{eq:conditional}
\pi_x(y) = \prob(Y=y \mid X=x), \quad y \in \YY. 
\end{equation}
For example, $\FTE$'s final reported probabilities on election day were $\pi_x(\clinton)=0.72$ and $\pi_x(\trump)=0.28$.
It is not our intention here to question the modeling decisions that produced these numbers. 
In fact, Silver's analysis seems to be based on too many {\em ad hoc}, inarticulable assumptions to evaluate in any scientific or systematic way, e.g., Silver's weighting scheme for certain polls based on his own subjective grading of their presumed accuracy.  Instead, we ask some more basic questions: {\em What do these numbers mean?} {\em What do they tell us about the election outcome?}

\subsection{Probabilities and their interpretation}
\label{SS:prob}

In the terminology of classical debates over probability, are these election probabilities to be interpreted as frequencies, subjective degrees of belief, betting quotients, or something else?  It turns out that there are elements of each interpretation, at least in the predictions from $\FTE$.  The discussion that follows highlights the difficulty in interpreting these probabilities from these classical points of view, and we offer a first look at an alternative approach that may be fruitful.  

\vspace{-3mm} \paragraph{Objective probabilities and frequencies.}  
In the case of $\FTE$, this 72\% figure represents the {\em frequency} of times Clinton won in a large number of simulations from their mathematical model.
In this sense, the probabilities are frequencies, but how do we interpret such a frequency in the context of the real world?
The frequency of 72\% derived from the $\FTE$ simulations is a synthetic frequency in that it is calculated by repeatedly sampling from an idealized model for the election outcome.
It does not represent a frequency in the same sense as frequencies obtained by tossing a physical coin a large number of times.
And, of course, the concept of repeating the election a large number of times is a fantasy.

Nevertheless, some of the post-election confusion stemmed from a widespread feeling in the general public that prediction probabilities correspond to some objective reality.  
For example, $\FTE$'s final pre-election report hints at a real-world interpretation by stating Clinton's ``chance of winning'' at 72\% \cite{Silver1108}; and {\em The New York Times} made a similar statement, giving Clinton an 85\% ``chance to win'' \cite{Upshot1108}.   There is a tendency to characterize Trump's victory as ``improbable'' based on the low percentages he was given under the models by $\FTE$, Princeton Election Consortium, {\em The New York Times}, and many others, but such an assessment affords those predictions an objective interpretation that may not be warranted.
For example, suppose that an urn contains 100 balls---50 blue and 50 red---and a ball is to be chosen completely at random.  If we know the composition of the urn, then neither color is {\em improbable}.  On the other hand, if we don't know the composition of the urn, but we assume that 1 ball is blue and 99 are red, then we would assign ``low credence,'' say, to drawing a blue ball even though such an outcome is by no means improbable.

\vspace{-3mm} \paragraph{Subjective probabilities and calibration.}  

Interpreting the probabilities as objective frequencies 
fails to account for the sensitivity of these probabilities to subjective modeling decisions.
In fact, the probabilities reported by $\FTE$ and others are quite sensitive to subjective modeling assumptions; for example, on October 24th, Silver estimated Clinton's chance at 85\%, but also discussed how the estimate would increase to anywhere from 90\% to 99.8\% if certain assumptions were modified \cite{Silver1024}.

Instead of interpreting these percentages as the long run frequency that the specific event ``Clinton is elected president in 2016'' occurs in a large number of repeated trials of the election, we can possibly interpret the probabilities as frequencies relative to the mechanism by which the percentages were derived.
Specifically, we may interpret the 72\% probability for Clinton in terms of how often we can expect any arbitrary event $E$ to occur when assigned probability $72\%$ by the methods employed by $\FTE$.
If the class of events $E$ to which the $\FTE$ simulations assign a 72\% probability occur roughly 72\% of the time, then the 72\% probability may have a somewhat objective interpretation as a frequency in terms of the mechanism by which it was calculated, but not in terms of the event that Clinton will win.

With this interpretation, one can supplement the model specification step with checks of how well the specified model fits the available data \cite{rubin1984}, with the goal of calibrating the subjective probabilities to some kind of long-run frequencies.  A necessary consequence of this form of calibration is that 72\% predictions have to be wrong about 28\% of the time \cite{Gelman}.  
Further shortcomings of calibration are highlighted in \cite{Shafer}.

\subsection{A new perspective---validity}
\label{SS:valid}

Generally speaking, it seems safe to assume that consumers of the election predictions don't care about the various interpretations of probability, e.g., betting odds, frequencies, etc.  They just want the prediction to be correct.  Since predictions come before the election is decided, there is uncertainty in these predictions, and the only logical reason to put stock in a given prediction is if the prediction is, in some sense, unlikely to be wrong.  With this in mind, we propose the condition of predictive {\em validity}, which we describe informally as follows:
\begin{quote}\label{quote:valid}
{\em a prediction derived from a given method is {\em valid} if the probability, relative to the posited model, that the method gives a wrong prediction is not too large.}
\end{quote}
This notion is closely related to Fisher's logic of statistical inference and Occam's razor: given the model, if the prediction derived from a particular method is incorrect with low probability, then there is no reason to doubt the prediction made in a particular case.  This goes beyond the idea of calibration discussed above in the sense that validity allows us to control the accuracy of our prediction, to mitigate our damages so to speak.  There are also some connections with the guarantees advocated by \cite{what.is, impred}.  

As we discuss further below, an important aspect of predictive validity is that it only requires that we can interpret small probabilities.  This feature aligns with the rationale of {\em Cournot's principle} \cite{Cournot},
according to which small (and therefore large) probabilities are the only interpretable or actionable probabilities.  For example, a moderate probability of 40\% for an event $E$ is likely too large for us to comfortably predict that $E$ will not occur, that is, to ``take action'' in predicting that its complement, $E^c$, will occur.  But we may well be comfortable to predict $E^c$ if the probability of $E$ were, say, $0.001$.  Our proposed approach also lends itself to accounting for uncertainty in the model, thereby avoiding an exaggeration of the quality of predictions, the ``overselling of precision'' discussed by some after the election \cite{NYTimes1}.  We discuss this further in Section~\ref{S:uncertainty}.

\subsection{Prediction sets and their validity}
\label{SS:sets}

By Cournot's principle, the only meaningful or actionable probabilities are the small ones.  In keeping with this principle, our proposal is to avoid directly interpreting those probabilities that are not small, and thus not easily interpretable.  Towards making this more formal, for a given prediction rule $x \mapsto \pi_x$ and a suitably small $\alpha > 0$ (see Proposition~\ref{prop:calibrated}), we define the {\em $\alpha$-level prediction set} as 
\begin{equation}
\label{eq:pred.set}
\Pi_x(\alpha) = \bigl\{ y \in \YY: \pi_x(y) > \alpha \bigr\}. 
\end{equation}
This set consists of all those election outcomes that are at least $\alpha$-probable according to the given prediction rule and the given polling data $x$.  We propose to report the prediction set $\Pi_x(\alpha)$ instead of committing to choose a single candidate and then potentially struggling to interpret the magnitudes of the individual prediction probabilities when making this choice.   By working with the set in \eqref{eq:pred.set}, we only commit ourselves to interpreting the probabilities with small magnitudes, i.e., we only rule out a candidate if his/her prediction probability is less than the specified threshold $\alpha$. 

For the 2016 election, the reported probabilities for Clinton were 0.72, 0.91, 0.98, and 0.99 for $\FTE$, {\em The New York Times}, {\em The Huffington Post}, and the Princeton Election Consortium, respectively.  Taking these probabilities as the values for $\pi_x(\clinton)$, if $\alpha=0.05$, then the resulting prediction sets are $\{\clinton, \trump\}$ for $\FTE$ and {\em The New York Times}, and $\{\clinton\}$ for the two others.

Our proposed approach that predicts with a set of candidates is more conservative, in general, than an approach that is obligated to choose a single candidate.  But if the polls are not informative enough to distinguish two or more candidates, then the analysis is prudent to avoid such a potentially misleading distinction.  Indeed, an important point that came up in the post-election discussions is that the predictions that heavily favored Clinton may have unintentionally caused some Clinton supporters not to vote.  Our proposal to make specific predictions only when the data makes it a clear decision helps to avoid the undesirable situation where the prediction itself affects the outcome.  
We note, however, that even under the suggested approach, the analysis by the Princeton Election Consortium and others that assigned probabilities of 98\% or higher to Clinton would still have given the wrong predictions, recalling the importance of a sound model in arriving at the probabilities $\pi_x$.

As we demonstrate below, under mild conditions, the $\alpha$-level prediction set $\Pi_x(\alpha)$ controls the probability of making a wrong prediction at $\alpha$.  More formally, we say that prediction based on a $\alpha$-level prediction set $\Pi_x(\alpha)$ defined in \eqref{eq:pred.set} is {\em valid} at level $\alpha$, relative to the specified model $\prob$ for $(X,Y)$, if 
\begin{equation}
\label{eq:valid}
\prob\{ \Pi_X(\alpha) \not\ni Y\} \leq \alpha.
\end{equation}
In words, the prediction is valid at level $\alpha$ if the $\prob$-probability of making a wrong prediction with the set $\Pi_X(\alpha)$ is no more than the specified $\alpha$.
Compare this to the informal notion of validity introduced in Section~\ref{SS:valid} above.  Proposition~\ref{prop:stochastic.order} gives a sufficient condition for \eqref{eq:valid} that is easy to arrange.  

\begin{prop}
\label{prop:stochastic.order}
A prediction set $\Pi_X(\alpha)$ as defined in \eqref{eq:pred.set} is valid in the sense of \eqref{eq:valid} if the corresponding $x \mapsto \pi_x$ satisfies 
\begin{equation}
\label{eq:stochastic.order}
\prob\{\pi_X(Y) \leq \alpha\} \leq \alpha.
\end{equation}
\end{prop}

\begin{proof}
By definition of $\Pi_X(\alpha)$, 
\[ \prob\{\Pi_X(\alpha) \not\ni Y\} = \prob\{\pi_X(Y) \leq \alpha\}, \]
so the validity result \eqref{eq:valid} follows immediately from \eqref{eq:stochastic.order}.  
\end{proof}

We note that the setup is very general---any choice of $x \mapsto \pi_x$ that satisfies \eqref{eq:stochastic.order} would lead to valid prediction---but \eqref{eq:stochastic.order} automatically holds for an interval of $\alpha$ values when $\pi_x$ is the conditional probability \eqref{eq:conditional} based on the model $\prob$.  

\begin{prop}
\label{prop:calibrated}
Let $\pi_x$ be the $\prob$-conditional probability in \eqref{eq:conditional}.  Then there exists $A > 0$ such that \eqref{eq:stochastic.order} holds for all $\alpha \in (0,A)$.   In particular, $A$ can be taken as 
\begin{equation}
\label{eq:A}
A = \min_x \pi_x^{(2)},
\end{equation}
where $\pi_x^{(2)}$ is the second smallest among the strictly positive $\pi_x(y)$ values, for the given $x$, as $y$ ranges over the (finite) set $\YY$. 
\end{prop}

\begin{proof}
Using iterated expectation and the definition of $\pi_x$, we get 
\begin{align*}
\prob\{ \pi_X(Y) \leq \alpha\} & = \E\{1_{\pi_X(Y) \leq \alpha}\} \\
& = \E\Bigl\{ \sum_{y \in \YY} 1_{\pi_X(y) \leq \alpha} \, \prob(Y=y \mid X) \Bigr\} \\
& =\E\Bigl\{\sum_{y \in \YY} 1_{\pi_X(y)\leq\alpha}\, \pi_X(y) \Bigr\}
\end{align*}
If $\alpha < A$, then the events $\{\pi_X(y) \leq \alpha\}$, $y \in \YY$, are mutually exclusive, so at most one term in the above summation is positive.  Since each term is no more than $\alpha$, this implies that the summation itself is no more than $\alpha$, hence \eqref{eq:stochastic.order}.
\end{proof}

Note that if $N=2$ in Proposition \ref{prop:calibrated}, 
say, $\YY = \{\clinton, \trump\}$, then $A$ in \eqref{eq:A} reduces to 
\[ A = \min_x \max\{\pi_x(\clinton), \pi_x(\trump)\}. \]
Unless there are some counter-intuitive constraints on the Clinton/Trump probabilities, e.g., $\pi_x(\clinton) \leq \frac13$ for all values of $x$, we expect $A=\frac12$.  So, the interval $(0,A)$ would certainly be wide enough to accommodate traditionally ``small'' values, such as $\alpha=0.05$.  

The above result is conservative in that the conclusion \eqref{eq:stochastic.order} may hold for a wider range of $\alpha$ than that given by $A$ in \eqref{eq:A}.  In fact, in some cases, \eqref{eq:stochastic.order} holds for all $\alpha \in (0,1)$, which amounts to $\pi_X(Y)$ being stochastically no smaller than a random variable having a uniform distribution on $(0,1)$; see Figure~\ref{fig:cdf}.  There are cases, however, where \eqref{eq:stochastic.order} does not hold for all $\alpha \in (0,1)$.  For an extreme example, in the case $\YY = \{\clinton, \trump\}$, suppose that the polling data is completely uninformative, so that $\pi_x(\clinton) = \pi_x(\trump) = \frac12$ for all $x$.  Then we have $\prob\{\pi_X(Y) \leq \alpha\} = 0$ for all $\alpha < \frac12$ and $\prob\{\pi_X(Y) \leq \alpha\} = 1$ for all $\alpha \geq \frac12$, so that \eqref{eq:stochastic.order} holds for $\alpha \in (0,\frac12)$ but in a trivial way.

\subsection{Illustrative example}
\label{SS:toy1}

Here we consider an example that is simple enough that the necessary calculations can be done without much trouble but also relevant enough that it illustrates our main points.
We revisit this example in Section~\ref{SS:toy2} below.  

Assume that there are only two viable candidates, Clinton and Trump, so that $\YY = \{\clinton, \trump\}$ and $N=2$.  Start with a fixed sample size $n$, poll $n$ voters, and record the number $X$ that indicate support for Trump.  (For now, we assume that there is no non-response, so that $n-X$ of the polled voters indicate support for Clinton; we relax this assumption in Section~\ref{SS:toy2}.)  A simple choice is a binomial model, i.e., $X \sim \bin(n,\theta)$, with a default flat prior for $\theta$, resulting in a model for $X$ that is uniform on $\{0,1,\ldots,n\}$.  To complete our specification of $\prob$, the joint distribution for the polling data $X$ and the election outcome $Y$, it remains to describe the conditional distribution of $Y$, given $X=x$, in \eqref{eq:conditional}.  Here we propose a logistic-type model
\begin{equation}\label{eq:logistic} \pi_x(\trump) = \frac{\exp\{\lambda(\hat\theta - \frac12)\}}{1 + \exp\{\lambda(\hat\theta - \frac12)\}} \quad \text{and} \quad \pi_x(\clinton) = 1 - \pi_x(\trump), \end{equation}
 where $\hat\theta = x/n$ and $\lambda > 0$ is a specified constant.  Assuming that the poll is reasonably informative, if $\hat\theta$ is much larger than 0.5---even, say, 0.6---then the probability assigned to a Trump victory ought to be close to 1, so a relatively large value of $\lambda$ makes sense; for this illustration, we consider $\lambda=10$.  A plot of this $\pi_x(\trump)$ curve, as a function of $\hat\theta$, is shown in Figure~\ref{fig:prob}.  The distribution function $G$ of $\pi_X(Y)$, derived from the joint distribution of $(X,Y)$ under this model, is displayed in Figure~\ref{fig:cdf}; note the stochastic ordering from \eqref{eq:stochastic.order}.  Validity of the corresponding prediction set $\Pi_X(\alpha)$, for any $\alpha \in (0,1)$, then follows immediately from Propositions~\ref{prop:stochastic.order} and \ref{prop:calibrated}.  

\begin{figure}
\begin{center}
\scalebox{0.8}{\includegraphics{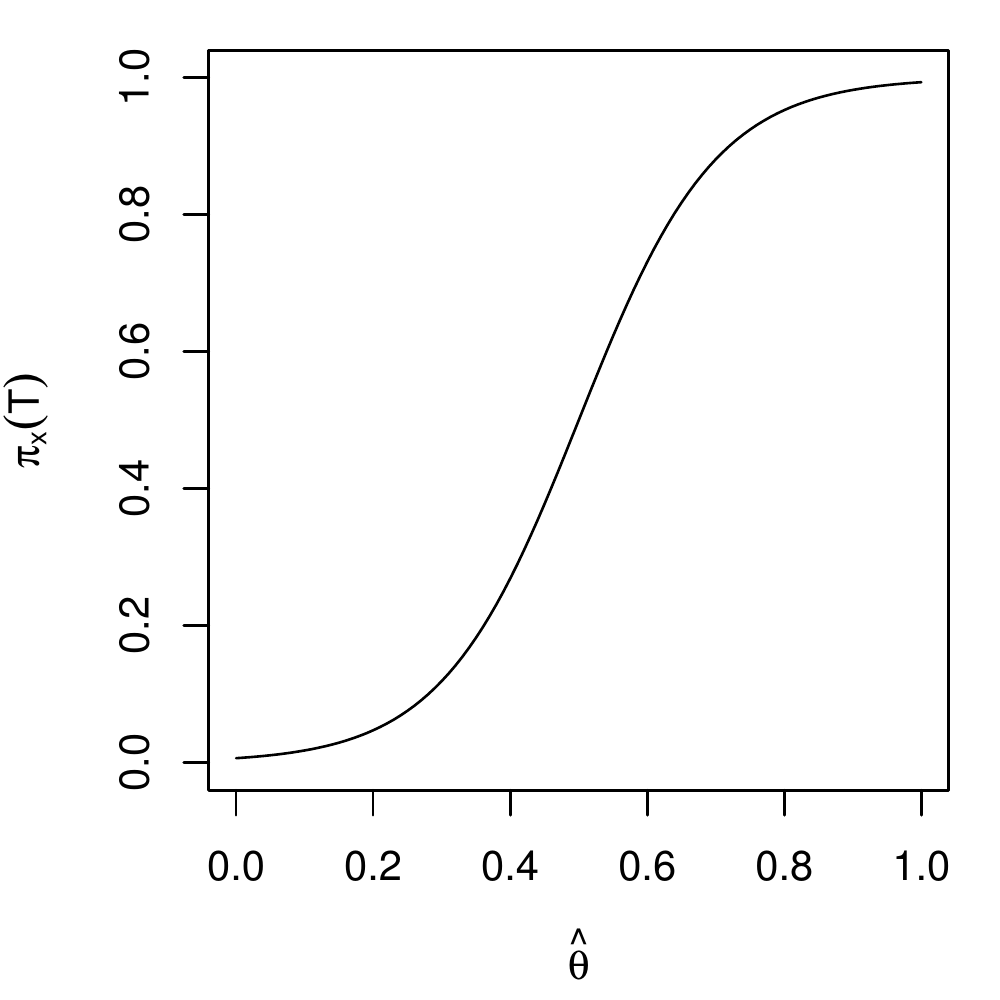}}
\end{center}
\caption{Plot of the probability $\pi_x(\trump)$ in Section~\ref{SS:toy1}, as a function of $\hat\theta=x/n$.}
\label{fig:prob}
\end{figure}

The prediction set $\Pi_x(\alpha)$, for small $\alpha$, would contain both $\clinton$ and $\trump$ unless 
\[ \hat\theta \leq \frac12 + \frac{1}{\lambda} \log\frac{\alpha}{1-\alpha} \quad \text{or} \quad \hat\theta \geq \frac12 + \frac{1}{\lambda} \log\frac{1-\alpha}{\alpha}. \]
For a numerical example, suppose that $X=470$ and we take the model in \eqref{eq:logistic} with $\lambda=10$.
According to the classical distribution theory, there is roughly a 3\% margin of error, so this seems to give an impression that Clinton has a significant lead. 
But under the model in \eqref{eq:logistic}, the prediction probabilities are 
\[\pi_x(\trump)=\frac{\exp\{10(0.47-0.50)\}}{1+\exp\{10(0.47-50)\}}=0.426\quad\text{and}\quad\pi_x(\clinton)=1-0.426=0.574,\]
giving both candidates an appreciable probability of winning.
For sure, no validity guarantees can be given for a prediction that Clinton wins based only on this available information.  Therefore, we argue that, based on this data and this model, the analyst should acknowledge that the race is too close to make a prediction, and report the set $\{\clinton, \trump\}$ based on this poll data.

\begin{figure}
\begin{center}
\scalebox{0.8}{\includegraphics{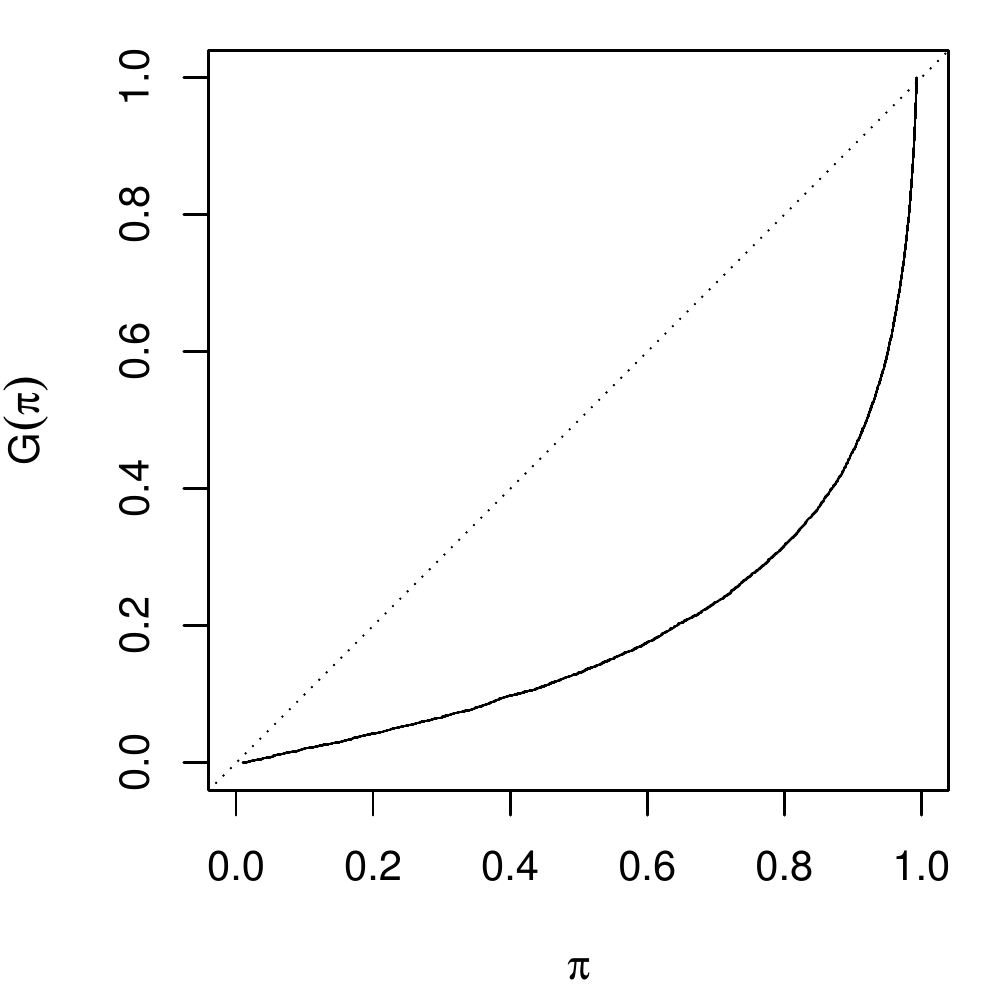}}
\end{center}
\caption{Distribution function $G(\pi) = \prob\{\pi_X(Y) \leq \pi\}$, for $\pi \in (0,1)$, derived from the joint distribution of $(X,Y)$ described in Section~\ref{SS:toy1}.}
\label{fig:cdf}
\end{figure}

\section{Predictions accounting for uncertainty}
\label{S:uncertainty}

\subsection{Probabilities to plausibilities}
\label{SS:plaus}

A serious concern about the use of (conditional) probabilities $\pi_x(y)$ in \eqref{eq:conditional} for prediction is that these only describe the {\em variability} of the outcome $Y$, given $X=x$, with respect to a given model $\prob$.  They do not reflect the {\em uncertainty} in the choice of $\prob$ itself.  Since the outcome $Y$ is usually binary in the election context, variability is difficult to interpret and arguably less important than model uncertainty.  

For example, recall Silver's October 24th acknowledgment \cite{Silver1024} that his estimates for Clinton could change from 85\% up to 99.8\% depending on the choice of model.  But then he ultimately reports an estimate that does not actually reflect this uncertainty.  An important lesson from this past election is that model uncertainty can be substantial, so it should somehow be incorporated into the predictions.  

A realistic situation is one where subject-area knowledge is used to narrow down to a collection $\probset$ of reasonable models $\prob$, but unable to identify a single best model.  That is, we ``don't know'' which model within this collection is most appropriate.  For example, in this past election, there was speculation of a ``hidden'' Trump vote, and among the collection of models that differ in the nature and extent to which the ``hidden'' vote influences the outcome, we don't know which is most appropriate; see Section~\ref{SS:toy2}.  Silver and others chose models that ignored the possible ``hidden'' Trump vote, and it turns out that this may have been a poor choice.  As discussed above, model choice is a subjective decision, so the only thing {\em wrong} with the choices made by Silver and others is in not accounting for their uncertainty in this choice. 

To account for uncertainty, we propose an adjustment to the output reported for prediction.  Specifically, given a collection $\probset$ of candidate models, we recommend reporting features of the set function 
\begin{equation}
\label{eq:pi.sup}
\upi_x(B) = \sup_{\prob \in \probset} \prob(Y \in B \mid X=x), \quad B \subseteq \YY, 
\end{equation}
in particular, 
\[ \upi_x(y) = \sup_{\prob \in \probset} \prob(Y=y \mid X=x), \quad y \in \YY. \]
In Section \ref{SS:interpret}, we discuss a number of important consequences in using the alternative prediction output in \eqref{eq:pi.sup}.  Here we simply point out that, since the supremum in $\upi_x(y)$ is taken separately for each $y \in \YY$, the resulting $\upi_x$ need not be a probability distribution.  For example, in the case where $\YY = \{\clinton, \trump\}$, it may be that $\upi_x(\clinton) + \upi_x(\trump) > 1$.  Therefore, since $\upi_x(y)$ may not be probabilities, we instead call them {\em plausibilities}, following Shafer's terminology \cite{Shafer1976}.  
In this terminology, an outcome $Y=y$ is {\em plausible} given $X=x$ as long as $Y=y$ has a large enough probability under {\em some} candidate model $\prob\in\probset$.
Conversely, an outcome is {\em implausible} if it has a small probability under all candidate models $\prob\in\probset$.

\subsection{Interpreting plausibilities}
\label{SS:interpret}

Our motivation for shifting from probabilities in Section~\ref{SS:prob} to plausibilities in Section~\ref{SS:plaus} is to account for model uncertainty in the prediction output.  Our proposal in \eqref{eq:pi.sup} satisfies this objective for two main reasons:
\begin{itemize}\addtolength{\itemsep}{-0.3\baselineskip}
\item Evaluating $\upi_x(y)$ requires the data analyst to be explicit about the class of models being entertained.  Compare this to standard practice where several models are considered but the results of only one cherry-picked analysis are reported \cite{gelman.lokin.2014}. 
\item Maximizing over $\probset$ is consistent with our general understanding (via Cournot's principle) that only small probability values are actionable.  In particular, when the model itself is uncertain---meaning that different candidate models in $\probset$ assign vastly different probabilities to a particular outcome---then we should be less inclined to take action based on the assessment of the more favorable model.  Taking the supremum in \eqref{eq:pi.sup} has the effect of pushing some borderline-actionable probabilities into the unactionable category.  
\end{itemize}

Plausibilities are surely less familiar and arguably more complicated than probabilities, so some additional care is needed to interpret them.  First, any set-function $p$ on the power set of $\YY$ that satisfies $p(\varnothing) = 0$, $p(B) \geq 0$, and $p(B_1) + p(B_2) \geq p(B_1 \cup B_2)$ is a plausibility function.  It is easy to check that the function $\upi_x$ in \eqref{eq:pi.sup} meets these requirements.  Towards an interpretation, it helps to regard $\upi_x$ it as an {\em upper probability} \cite{dempster1967}; that is, $\upi_x(B)$ is the maximal probability assigned to $B$ as $\prob$ varies over the collection $\probset$.  Conversely, there is a corresponding {\em lower probability} 
\[ \lpi_x(B) = 1 - \upi_x(B^c), \quad B \subseteq \YY, \]
or {\em belief function} \cite{Shafer1976}, obtained by replacing supremum in \eqref{eq:pi.sup} with infimum.  

The lower probabilities $\lpi_x$ and, therefore, also upper probabilities $\upi_x$ admit a clear interpretation in terms of bets \cite{shafer2011, williams1978}.  If our lower probability on an event $B\subseteq\YY$ is given by $\lpi_x(B)$ and we are offered the opportunity to bet on $B$, then we would accept any price up to $\lpi_x(B)$ since such a bet is coherent with respect to every conceivable model $\prob$ in $\probset$.  This interpretation of $\lpi_x(B)$ in terms of betting is more conservative than the ordinary Bayesian arguments for coherence in terms of two-sided bets.  In the conventional setting, we are offered a bet at price $p_B$ on $B$ and $1-p_B$ on $B^c$.
Note that by the requirement $\upi_x(B)+\upi_x(B^c)\geq1$, we have $\lpi_x(B)+\lpi_x(B^c)\leq 1$, so that if $p_B\in(\lpi_x(B),\upi_x(B))$, then we would not accept a bet on either $B$ or $B^c$.

Another interesting feature of the upper and lower probabilities is that the difference
\[ \upi_x(B) - \lpi_x(B)\]
measures model uncertainty about $B$, 
what Dempster \cite{dempster2008} refers to as a ``don't know'' probability, which measures the amount of uncertainty about $B$ in the model $\probset$.  The greater this difference, the more uncertainty there is about how to assign probabilities to the event $B$, and hence the more we ``don't know'' about the proper assessment of $B$.  In the election context, where we are interested primarily in the probabilities assigned to particular candidates, we have 
\[ \upi_x(y) = \sup_{\prob \in \probset} \prob(Y=y \mid X=x) \quad \text{and} \quad \lpi_x(y) = \inf_{\prob \in \probset} \prob(Y = y \mid X=x), \]
and the uncertainty about the model is reflected in the fact that the difference $\upi_x(y) - \lpi_x(y)$ is positive at least for some $y \in \YY$.  Under the view that only small probabilities are interpretable, the most meaningful/useful output is $\upi_x(y)$, $y \in \YY$, since having $\upi_x(y)$ small means that the event $Y=y$ has small probability under all candidate models $\prob\in\probset$, and is thus actionable even in the presence of model uncertainty.  

Finally, our choice to work with the plausibility or upper probability $\upi_x$ is consistent with the proposed notion of validity and the corresponding prediction sets.  Indeed, the following proposition says that, for suitable $\alpha$, the $\alpha$-prediction set $\overline{\Pi}_x(\alpha)$, using the plausibility $\upi_x(y)$ in place of the probability $\pi_x(y)$, satisfies the prediction validity condition \eqref{eq:valid} with respect to any $\prob$ in $\probset$.

\begin{prop}
\label{prop:calibrated.pl}
Suppose that there exists $\alpha > 0$ such that 
\[ \prob\{\pi_X(Y) \leq \alpha\} \leq \alpha \]
for all $\prob \in \probset$, where $\pi_x$ is the $\prob$-conditional probability.  Then, for that $\alpha$ value, the plausibility-based $\alpha$-prediction set 
\[ \overline{\Pi}_x(\alpha) = \{y \in \YY: \upi_x(y) > \alpha\} \]
satisfies 
\[ \sup_{\prob \in \probset} \prob\{\overline{\Pi}_X(\alpha) \not\ni Y\} \leq \alpha. \]
\end{prop}

\subsection{Illustrative example, cont.}
\label{SS:toy2}

To follow up on the example from Section~\ref{SS:toy1}, suppose now that, among the $n=1000$ polled individuals, 475 pledge support to Clinton, 425 pledged support to Trump, and 100 did not respond.  A naive approach to handling non-response is to assume that those who did not respond are no different, on average, than those who did respond; in statistics jargon, this corresponds to a {\em missing-at-random} assumption.  But based on the available information, there  is no justification for such a model assumption.  In fact, all available information indicates that these 100 individuals, in choosing not to respond to the poll, acted drastically differently from the other 900 who did respond, suggesting that their preferences may also exhibit different behavior.

Here, the added uncertainty is in which category---$\clinton$ or $\trump$---to put the 100 non-responses.  The missing-at-random model would split the 100 proportionally into the two groups, which boils down to ignoring them altogether.  Another model, extreme but consistent with the hypothesis of a ``hidden'' Trump vote would put all 100 non-responses with Trump; a similarly extreme model puts the 100 non-responses with Clinton, consistent with a ``hidden Clinton'' vote that was discussed by some \cite{WPost1106}.  There may be a whole spectrum of models in $\probset$, differing in how they handle the non-responses, but we have identified the two extremes, so we can evaluate the plausibilities by plugging the $\hat\theta$ corresponding to these extremes into the formula in Section~\ref{SS:toy1}.  That is, set 
\[ \hat\theta_{\text{lo}} = \frac{425 + 0}{1000} = 0.425 \quad \text{and} \quad \hat\theta_{\text{hi}} = \frac{425 + 100}{1000} = 0.525, \]
and evaluate the plausibilities (for $\lambda=10$)
\begin{align*}
\upi_x(\trump) & = \frac{\exp\{\lambda(\hat\theta_{\text{hi}} - \frac12)\}}{1 + \exp\{\lambda(\hat\theta_{\text{hi}} - \frac12)\}} = 0.562, \\
\upi_x(\clinton) & = 1 - \frac{\exp\{\lambda(\hat\theta_{\text{lo}} - \frac12)\}}{1 + \exp\{\lambda(\hat\theta_{\text{lo}} - \frac12)\}} = 0.679.
\end{align*}
This is equivalent to a pair of lower probabilities, 0.321 for $\trump$ and 0.438 for $\clinton$, along with an extra 0.241 of ``don't know.''  These results indicate that the information in the polling data is too weak to make a clear prediction about the winner.  Compare this to a naive analysis that ignores the non-response in assuming $\hat\theta=425/900=0.472$ and reports the pair of probabilities 
\[\pi_x(\trump)=\frac{\exp\{\lambda(0.472-0.50)\}}{1+\exp\{\lambda(0.472-0.50)\}}=0.430\quad\text{and}\quad\pi_x(\clinton)=1-\pi_x(\trump)=0.570.\]
This latter analysis gives an indication that Clinton has a semi-healthy lead, while ignoring that a critical and arguably unjustifiable missing-at-random assumption was made.

\section{Concluding remarks}\label{S:fail}

The field of data science and statistics has grown rapidly over the past decade in no small part because of the perceived success of data-driven methods in predicting the 2008 and 2012 presidential elections.  Writing for a general audience, Kass et al (2016) pay explicit acknowledgment to ``$\FTE$ for bringing statistics to the world (or at least to the media).''  But the same publicity which gained statistics and data science great fanfare after the 2008 and 2012 elections led to the public outcry after its perceived failure in the 2016 election.  Such is the risk that comes with being relevant, and statisticians ought to avoid the mistakes of the past when responding to these criticisms.

Statisticians are inclined to defend against these criticisms.  For example, the brazenness in the declaration on the {\em SimplyStatistics} blog that {\em Statistically speaking, Nate Silver, once again, got it right}, leaves many ordinary citizens scratching their heads: Silver was praised for correctly predicting the 2008 and 2012 elections; and now he's being praised for incorrectly predicting the 2016 election.  If Silver was correct in 2008 and 2012, when his predictions were right, and also correct in 2016, when his predictions were wrong, then is it possible for the analysis to ever be incorrect?  Though professional statisticians may dismiss such reasoning as naive, defending against these common sense refutations with jargon heavy arguments only arouses more suspicion.

To be clear, we are not claiming that those analyses were {\em wrong}, only that they were {\em ineffective}.  And by all accounts, the analyses {\em failed}.
They failed because the 2016 election was one of the most heavily covered news events in recent history, 
offering a major platform to publicize the power of data-driven techniques in addressing real world problems.
But that's OK.  Statisticians should not feel obligated to defend these failed predictions.  We should, instead---as urged by Jessica Utts \cite{utts}, 2016 President of the American Statistical Association---take the opportunity to clarify the purpose of statistical analyses in the first place.  Statistics is not a tool for predicting the future as much as it is of making sense of what insights can be reasonably gleaned from data.  Statistical analysis is conservative by its nature, and the limelight enjoyed by Nate Silver and other data scientists shouldn't distract from that.

As Yogi Berra famously said, ``It's tough to make predictions, especially about the future,''  
so even a perfectly sound analysis is not guaranteed to be right.  One of the things we can learn from the 2016 election is just how embarrassing incorrect predictions can be.  Since there is sure to be uncertainty in the 2020 election, we must be more careful about how we justify and communicate our conclusions to consumers.  In this paper, we have put forward some alternative ways to convert probabilities to valid predictions in a way that accounts for model uncertainties.  These are, by design, more conservative than traditional prediction strategies, especially when there is uncertainty about the model itself, but we show that it is  straightforward to arrange some bounds on the probability of making an erroneous prediction.

\bibliography{refs}
\bibliographystyle{abbrv}

\end{document}